\documentclass[runningheads]{llncs}
\usepackage{graphicx}
\usepackage{amssymb}
\usepackage{amsmath}
\usepackage{algorithm}
\usepackage[noend]{algorithmic}

\allowdisplaybreaks
\begin{document}
\title{Minimum-weight partitioning of a set with associated subsets%\thanks{Supported by organization x.}
}
%
%\titlerunning{}
\author{Yakov Zinder\inst{1}\orcidID{0000-0003-2024-8129} 
\and
Bertrand M.T. Lin\inst{2}\orcidID{0000-0003-0456-296X}
\and
Joanna Berli\'{n}ska\inst{3}\orcidID{0000-0003-2120-2595} 
}
\authorrunning{Y. Zinder et al.}
\institute{School of Mathematical and Physical Sciences, University of Technology Sydney, Australia \\
\email{Yakov.Zinder@uts.edu.au}
\and
Institute of Information Management, National Yang Ming Chiao Tung University, Hsinchu, Taiwan\\
\email{bmtlin@nycu.edu.tw}
\and
Faculty of Mathematics and Computer Science, Adam Mickiewicz University, Pozna\'n, Poland\\
\email{Joanna.Berlinska@amu.edu.pl} }
\maketitle              % typeset the header of the contribution
\begin{abstract}
The paper presents complexity results and performance guaranties for a family of approximation algorithms for an optimisation problem arising in software testing and manufacturing. The problem is formulated as a partitioning of a set where each element has an associated subset in another set, but  can also be viewed as a scheduling problem with infinitely large communication delay, precedence constraints in the form of a bipartite graph, and duplication. 

\end{abstract}
\section{Introduction}

Consider two finite non-empty sets $N$ and $M$ such that $N \cap M = \emptyset$ and each $i \in N\cup M$ has a positive integer weight $p_i$. Let each $i \in N$ be associated with some subset $M(i)$ of $M$ in such a way that
$\cup_{i \in N} M(i) = M$, and let $m$ be a positive integer, satisfying $m < |N|$. The considered problem requires to find a partition of $N$ into $m$ non-empty subsets $N^1$, ..., $N^m$ with the smallest  
\begin{equation}\label{f()}
 f(N^1, ..., N^m) = \max_{1 \le e \le m} \left( \sum_{i \in N^e}p_i + \sum_{j \in \cup_{i \in N^e} M(i)}p_j \right).
\end{equation}
In what follows, this problem will be referred to as Minimum-Weight Partitioning of a Set with Associated Subsets (MWPSAS).

Originally, the formulation and study of MWPSAS were triggered by the collaboration with a company that needed to assign a set $N$ of certain computer programs to $m < |N|$ identical workstations for the purpose of software testing. Each $i \in N$ is assigned to a workstation together with a set $M(i)$ of certain associated operations such as the installation of special libraries. Each $g \in N \cup \left(\cup_{i \in N} M(i)\right )$ consumes $p_g$ units of time, but if $i \in N$ and $j \in N$ are assigned to the same workstation and $M(i) \cap M(j) \neq \emptyset$, then the operations constituting $M(i) \cap M(j)$ need to be performed only once. It was necessary to find an assignment of programs to workstations (the partition $N^1, ..., N^m$ in the formulation above) with the smallest time needed for finishing the entire testing. A similar situation arises in manufacturing, when jobs require several tools, and a tool that is installed for processing one job can be used for processing any other job on the same machining center.

The study was also motivated by the obvious relation between MWPSAS and scheduling with communication delay -- an area of scheduling theory which has remained a subject of intensive research for several decades since the pioneer publications   
\cite{rayward1987uet} and \cite{papadimitriou1990towards}. For the purpose of this discussion, the problem of scheduling with communication delay can be stated as follows. A set of jobs $J$ is to be processed on $m > 1$ parallel identical machines subject to precedence constraints in the form of antisymmetric, antireflexive, and transitive partial binary relation on $J$. If $i \in J$ precedes $j\in J$ in this relation, denoted $i \rightarrow j$, then the processing of $j$ can commence only after the completion of $i$. Furthermore, if $i \rightarrow j$ and these jobs are processed on different machines, then the processing of $j$ can commence not earlier than $c$ time units after the completion of $i$. The constant $c$ is referred to as communication delay. Each job $i \in J$ can be processed by any machine and its processing time $p_i$  remains the same regardless of the machine to which this job is allocated. Each machine can process not more than one job at a time. It is necessary to minimise the time needed to complete all jobs in $J$, commonly denoted by $C_{max}$ and referred to as the makespan. 

Using the three-field notation, adopted in scheduling theory, this problem can be denoted $P|prec, c_{ij}=c|C_{max}$. If all instances of this scheduling problem have the same number of machines, then the first field changes from $P$, indicating that the number of parallel identical machines is part of the input and may vary from instance to instance, to $Pm$, indicating that the number of machines is the same for all instances. Often the partially ordered set of jobs is considered as an acyclic directed graph where $J$ is the set of nodes and there is an arc from $i$ to $j$ if and only if $i \rightarrow j$. If all instances have the same type of precedence constraints, then $prec$ in the second field changes respectively, for example, to $bipartite$ when only bipartite graphs are considered. Similarly, $p_i = 1$ and $c_{ij}=\infty$ indicate that all jobs in $J$ require one unit of machine time and that the communication delay is infinitely large, respectively. One of the techniques aimed at reducing the makespan by avoiding communication delay is duplication, which is the permission to transform the original acyclic directed graph into a new one by replacing some nodes by their several exact copies, i.e. if node $i$ is replaced by nodes $i_1, ..., i_k$, then for each $1 \le e \le k$, $p_{i_e} = p_i$ and node $i_e$ precedes or succeeds any node $j$ if and only if it is true for $i$. If duplication is allowed, this is indicated in the second field of the three-field notation by $dup$.   

It is easy to see that MWPSAS is equivalent to $P|bipartite, dup, c_{ij}=\infty|C_{\max}$ with the set of jobs $N\cup M$ and where $j \rightarrow i$ if and only if $i \in N$ and $j \in M(i)$. In terms of scheduling with communication delay, this paper contributes the following results: (a) a proof of the NP-hardness in the strong sense of $P2|bipartite, dup,$ $c_{ij}=\infty, p_i=1|C_{\max}$; (b) a proof of the NP-hardness in the strong sense of $P|2in\textrm{-}tree,$ $dup, c_{ij}=\infty, p_i=1|C_{\max}$ where $2in\textrm{-}tree$ indicates that the directed acyclic graph is a collection of in-trees in which any path cannot contain more than two nodes; (c) a proof of the NP-hardness in the strong sense of $P|2out\textrm{-}tree, dup, c_{ij}=\infty, p_i=1|C_{\max}$ where $2out\textrm{-}tree$ indicates that the directed acyclic graph is a collection of out-trees in which any path cannot contain more than two nodes; (d) a performance guaranty (an upper bound on the deviation from the optimal makespan) for a broad class of approximation algorithms for $P|bipartite, dup, c_{ij}=\infty|C_{\max}$ which contains an algorithm with the best possible performance guaranty for $P|2out\textrm{-}tree, dup, c_{ij}=\infty, p_i=1|C_{\max}$.

\section{Computational complexity}
An instance of the decision version of MWPSAS requires an answer to the question: for a given positive integer $C$, does there exist a partition of $N$ into $m$ non-empty subsets $N^1$, ..., $N^m$ for which $f(N^1, ..., N^m)$ does not exceed $C$.

The proof below is a reduction from the Clique Problem (CLIQUE) that is NP-complete in the strong sense \cite{garey1979computers} and can be stated as follows: 

\vspace{0.2cm}
\parbox{12cm}{
CLIQUE Given an integer $k > 1$ and a graph $G(V,W)$, where $V$ is the set of nodes and  $W$ is the set of edges, such that
\begin{equation}\label{|W|<=}
   k < |V| \hspace{0.3cm}\mbox{and}\hspace{0.3cm}  \frac{k^2-k}{2} < |W| .
\end{equation}
Does $G(V,W)$ contain $k$ nodes that induce a subgraph where any two nodes are linked by an edge (such a subgraph is called complete or a $k$-clique)?
}

\noindent Observe that the largest number of edges in a graph with $k$ nodes is 
$
 \displaystyle\frac{k^2-k}{2},
$
and a graph with $k$ nodes has this number of edges if and only if it is complete (is a $k$-clique). 
\begin{theorem}\label{theorem: m=2}
 MWPSAS is NP-hard in the strong sense even when $m=2$ and $p_i=1$ for all $i \in N\cup M$.
\end{theorem}
\begin{proof}
 Let an integer $k > 1$ and a graph $G(V,W)$ be an instance of the CLIQUE problem, where $V=\{v_1, ..., v_n\}$ is the set of nodes and $W$ is the set of edges. The edge that links nodes $v_i$ and $v_j$ will be denoted by $\{v_i, v_j\}$. The corresponding instance of the decision version of MWPSAS is constructed as follows. The set $N$ is a union of three disjoint sets $W$, $T$, and $T_0$, where
 \begin{equation}\label{|T|}
  |T| = \frac{k^2+k}{2}\hspace{0.3cm}\mbox{ and }\hspace{0.3cm} |T_0| = 1;
 \end{equation}
 the set $M$ is a union of two disjoint sets $V$ and $S$, where
\begin{equation}\label{|S|}
  |S| = n + |W| - \frac{k^2-k}{2} - 1;
 \end{equation}
for each $a \in N$,
\begin{equation}\label{N(a)}
 M(a) = \left\{
 		\begin{array}{cl}
			S & \mbox{ if }  \{a\} = T_0 \\
			\{v\} \cup \{u\} & \mbox{ if } a = \{v, u\} \in W\\
			V & \mbox{ if } a \in T
		\end{array}
 	  \right. ;
\end{equation}
and finally
 \[ 
  C =  n + |W| + k, \hspace{0.3cm} m = 2, \hspace{0.3cm}\mbox{ and }\hspace{0.3cm} p_i =1\hspace{0.3cm}\mbox{ for all }\hspace{0.3cm} i\in N\cup M.  
 \]

Suppose that there exists a partition of $N$ into 2 non-empty subsets $N^1$ and $N^2$ such that  
\begin{equation}\label{f<C}
 f(N^1,  N^2) \le C = n + |W| + k.
\end{equation}
Without loss of generality assume that $T_0 \subseteq N^1$. Then, by virtue of (\ref{N(a)}),  $S \subseteq \cup_{i \in N^1} M(i)$.  Furthermore, this assumption leads to the inequality 
\begin{equation}\label{TcapN2}
 T \cap N^2 \neq \emptyset,
\end{equation}
 because otherwise $T \subset N^1$, which,  by virtue of (\ref{N(a)}), implies $V  \subset \cup_{i \in N^1} M(i)$, and consequently,  
\[
 f(N^1, N^2) \ge \sum_{i \in N^1}p_i + \sum_{j \in \cup_{i \in N^1} M(i)}p_j \ge  |T_0| + |T| + |S| + |V| 
\]
\[ 
 = 1 + \frac{k^2+k}{2} + n + |W| - \frac{k^2-k}{2} -1 + n > C,
\]
which contradicts (\ref{f<C}). This, in turn, implies that $T \cap N^1 = \emptyset$, 
because otherwise, by virtue of (\ref{N(a)}) and (\ref{TcapN2}),  
\[
 V  \subset \cup_{i \in N^1} M(i) \hspace{0.3cm}\mbox{ and }\hspace{0.3cm} V  \subseteq \cup_{i \in N^2} M(i),
\] 
and consequently, taking into account  (\ref{|W|<=}),
\[
 f(N^1, N^2)  \ge \frac{2|V| + |S| + |W| + |T| +|T_0|}{2}
\]
\[
 = n +  \frac{n}{2} + \frac{|W|}{2} - \frac{k^2-k}{4}  - \frac{1}{2} + \frac{|W|}{2} + \frac{k^2 + k}{4} + \frac{1}{2} > C,
\]
which contradicts (\ref{f<C}). 

Summarising the above, 
\begin{equation}\label{sum}
  T_0 \subseteq N^1, \hspace{0.3cm} S \subseteq \cup_{i \in N^1} M(i), \hspace{0.3cm}  T \subseteq N^2, \hspace{0.3cm}  V  \subseteq \cup_{i \in N^2} M(i), 
\end{equation}
and hence, 
\[
 |W \cap N^2| \le C - |V| - |T| = n +|W| + k - n  - \frac{k^2+k}{2}  =  |W| - \frac{k^2-k}{2},
\]
which, since $|W| - |W \cap N^2| = |W \cap N^1|$, gives
\begin{equation}\label{T1inN2}
 \frac{k^2-k}{2} \le  |W \cap N^1|.
\end{equation}
Furthermore, by (\ref{f<C}) and (\ref{sum}), 
\[
  | \cup_{i \in N^1} M(i)\cap V| + |W \cap N^1| \le C - |S| - |T_0|
 \]
\[
 = n +|W| + k  - \left (n + |W| - \frac{k^2-k}{2} - 1\right) - 1
 = \frac{k^2-k}{2} + k;
\]
and by adding 
\[
  | \cup_{i \in N^1} M(i)\cap V| + |W \cap N^1|  \le \frac{k^2-k}{2} + k
\]
and (\ref{T1inN2}), 
\begin{equation}\label{S1inN1}
 | \cup_{i \in N^1} M(i)\cap V| \le k.
\end{equation}
On the other hand, (\ref{N(a)}) implies that 
\begin{equation}\label{WN1}
  |W\cap N^1|  \le \frac{| \cup_{i \in N^1} M(i)\cap V|^2 - | \cup_{i \in N^1} M(i)\cap V|}{2} ,
\end{equation}
where the right-hand side of this inequality is the largest number of edges in a graph with $| \cup_{i \in N^1} M(i)\cap V|$ nodes. Taking into account (\ref{T1inN2}),  (\ref{S1inN1}), and (\ref{WN1}), 
\begin{equation}\label{WN1<k}
 |W \cap N^1| =  \frac{k^2-k}{2},
\end{equation}
which, in turn, together with (\ref{S1inN1}), gives $| \cup_{i \in N^1} M(i)\cap V| = k$, because the right-hand side of (\ref{WN1<k})  is the largest number of edges in a graph with $k$ nodes. Hence, the graph induced by the nodes in $ \cup_{i \in N^1} M(i)\cap V$ is a $k$-clique.

Conversely, assume that graph $G(V, W)$ has a $k$-clique. Without loss of generality, let $v_1$, ..., $v_k$ be the nodes in this clique. Consider the partition of $N$ into $N^1$ and $N^2$ where
\begin{equation}\label{N1=T0 and k}
 N^1 = T_0 \cup \{\{v_i, v_j\}: i < j \le k\}. 
\end{equation}
Then, according to (\ref{N(a)}), 
\[
 \cup_{i \in N^1} M(i) = \{v_1, ..., v_k\} \cup S,
\]
and taking into account (\ref{|T|}) and (\ref{|S|}),
\[
 \sum_{i \in N^1}p_i + \sum_{j \in \cup_{i \in N^1} M(i)}p_j = 1 + \frac{k^2-k}{2} + k  
 + n + |W| - \frac{k^2-k}{2} - 1 = C.
\]

Furthermore, (\ref{N1=T0 and k}) implies that 
\begin{equation}\label{N2=T and W}
  N^2 = T \cup (W \setminus \{\{v_i, v_j\}: i < j \le k\}),
\end{equation}
which, in turn, by virtue of (\ref{N(a)}), implies
\[
 \cup_{i \in N^2} M(i) = V.
\]
Consequently, taking into account (\ref{|T|}),
\[
 \sum_{i \in N^2}p_i + \sum_{j \in \cup_{i \in N^2} M(i)}p_j = |T| + |W \setminus \{\{v_i, v_j\}: i < j \le k\}| +|V| 
\]
\[
 = \frac{k^2+k}{2} + |W| - \frac{k^2-k}{2} + n  = C.
\]
Hence, the partition of $N$, defined by  (\ref{N1=T0 and k}) and (\ref{N2=T and W}), satisfies  (\ref{f<C}).
\qed
\end{proof}

If $m$ is part of the input, then the NP-hardness in the strong sense can be established even for very restricted cases of MWPSAS. Two such cases, \mbox{MWPSAS\_M1} and MWPSAS\_N1, are considered below. For any $j \in M$, denote by $N(j)$ the set of all $i \in N$ such that $j \in M(i)$.
\begin{itemize}
\item MWPSAS\_M1 is a particular case of MWPSAS where 
	\begin{itemize}	
		\item $p_i=1$ for all $i \in N\cup M$, 
		\item $|M(i)| = 1$ for all $i \in N$,
		\item $N(i) \cap N(j) = \emptyset$ for each $\{i, j\} \subseteq M$;
	\end{itemize}	
\item MWPSAS\_N1  is a particular case of MWPSAS where 
	\begin{itemize}	
		\item $p_i=1$ for all $i \in N\cup M$, 
		\item $|N(j)| = 1$ for all $j \in M$,
		\item $M(i) \cap M(j) = \emptyset$ for each $\{i, j\} \subseteq N$.
	\end{itemize}	
\end{itemize}

The proofs of the following two theorems are reductions from the 3-partition problem which will be referred to as 3-PARTITION. 3-PARTITION is NP-complete in the strong sense  \cite{garey1979computers}  and can be stated as follows: 

\vspace{0.2cm}
\parbox{12cm}{
3-PARTITION Given $3r$ integers $a_1$, \ldots, $a_{3r}$, each greater than one, and an integer $B$ such that 
\begin{equation}\label{3-PART}
 \sum_{k = 1}^{3r}a_k = r B
\end{equation}
and, for each $k$, 
\begin{equation}\label{3-PARTaB}
 \frac{B}{4} < a_k < \frac{B}{2}. 
\end{equation}
Does there exist a partition of the set $\{1, 2, \ldots, 3r\}$ into $r$ subsets $Z_1$, \ldots, $Z_r$ such that, for each $Z_e$,
\begin{equation}\label{sum a=B}
 \sum_{k \in Z_e} a_k = B?
\end{equation}
}

\noindent Observe that (\ref{3-PARTaB}) implies that if the desired partition exists, the cardinality of each $Z_k$ is equal to three. 

\begin{theorem}\label{MWPSASM1}
MWPSAS\_M1 is NP-hard in the strong sense.
\end{theorem}

\begin{proof} 
Consider an instance of 3-PARTITION, i.e. $3r$ integers $a_1$, \ldots, $a_{3r}$, each greater than one, and an integer $B$, satisfying (\ref{3-PART}) and (\ref{3-PARTaB}). The corresponding instance of the decision version of the  MWPSAS\_M1 problem is constructed as follows. Set
\begin{itemize}
 \item $m = r$ and $C = B$; 
 \item $M = \{j_1, ..., j_{3r}\}$ and $N = \cup_{k=1}^{3r}A_k$ where $A_1, ..., A_{3r}$ are  disjoint sets such that $|A_k| = a_k -1$  for each $1 \le k \le 3r$; 
 \item $N(j_k) = A_k$ for each $1 \le k \le 3r$, and $p_i = 1$ for all $i \in N \cup M$.
\end{itemize} 
Observe that for this instance of the decision version of MWPSAS\_M1
\begin{equation}\label{thm2: sumM1}
 \sum_{i \in N \cup M} p_i = \sum_{k=1}^{3r} \sum_{g \in A_k \cup \{j_k\}}p_g = \sum_{k=1}^{3r} a_k = rB = mB. 
\end{equation} 
 
Suppose that there exists a partition of $N$ into $m$ non-empty subsets $N^1$, ..., $N^m$ such that  
\begin{equation}\label{f<C=B}
 f(N^1, ..., N^m) \le C = B.
\end{equation}
Then, taking into account (\ref{f()}), for each $1 \le e \le m$,
\[
 B \ge \sum_{i \in N^e}p_i + \sum_{j \in \cup_{i \in N^e} M(i)}p_j,
\]
which together with (\ref{thm2: sumM1}) gives
\[
 mB \ge \sum_{e=1}^m \left(\sum_{i \in N^e}p_i + \sum_{j \in \cup_{i \in N^e} M(i)}p_j\right)  \ge \sum_{i \in N \cup M} p_i = mB,
\] 
and consequently,
\begin{equation}\label{thm2: sum()=sum pi}
 \sum_{e=1}^m \left(\sum_{i \in N^e}p_i + \sum_{j \in \cup_{i \in N^e} M(i)}p_j\right)  = \sum_{i \in N \cup M} p_i
\end{equation}
and
\begin{equation}\label{thm2: ()=B}
 \sum_{i \in N^e}p_i + \sum_{j \in \cup_{i \in N^e} M(i)}p_j = B \hspace{0.5cm}\mbox{ for each } 1\le e \le m.
\end{equation}

For each $1 \le k \le 3m$, there exists a unique $c$ such that $j_k \in \cup_{i \in N^c}M(i)$, because otherwise $p_{j_k}$ will appear in the left-hand side of (\ref{thm2: sum()=sum pi}) more than once, which contradicts this equality. This implies that if index $e$ is not equal to this $c$, then $A_k \cap N^e = \emptyset$, and consequently $A_k \subseteq N^c$. For each $1 \le e \le m$, denote by $K_e$ the set of all $k$ such that $j_k \in \cup_{i \in N^e}M(i)$. Then, taking into account (\ref{thm2: ()=B}),
\[
 B = \sum_{i \in N^e}p_i + \sum_{j \in \cup_{i \in N^e} M(i)}p_j = \sum_{k \in K_e} \sum_{g \in A_k \cup \{j_k\}}p_g = \sum_{k \in K_e}a_k.
\]
So, $Z_1 = K_1$, ..., $Z_r = K_r$ is a required partition of the set $\{1, 2, ..., 3r\}$.
 
Conversely, let $Z_1$, ..., $Z_r$ be a partition of the set $\{1, 2, \ldots, 3r\}$ such that each $Z_e$ satisfies (\ref{sum a=B}). Then, taking into account that $m=r$,  
\[
 N^1 = \cup_{k \in Z_1}A_k, \hspace{0.5cm}..., \hspace{0.5cm}N^m = \cup_{k \in Z_r}A_k
\]
is a partition of $N$ which satisfies (\ref{f<C=B}) because, for each $1 \le e \le m$,
\[
 \sum_{i \in N^e}p_i + \sum_{j \in \cup_{i \in N^e} M(i)}p_j = 
 \sum_{k \in Z_e}|A_k| + |Z_e| = 
\]
\[
 \sum_{k \in Z_e}(a_k - 1) + |Z_e| = 
 \sum_{k \in Z_e}a_k = B. 
\]
\qed 
\end{proof}

\begin{theorem}\label{MWPSASN1}
MWPSAS\_N1 is NP-hard in the strong sense.
\end{theorem}
\begin{proof}
Consider an instance of 3-PARTITION, i.e. $3r$ integers $a_1$, \ldots, $a_{3r}$, each greater than one, and an integer $B$, satisfying (\ref{3-PART}) and (\ref{3-PARTaB}). The corresponding instance of the decision version of MWPSAS\_N1 is constructed as follows. Set
\begin{itemize}
 \item $m = r$ and $C = B$; 
 \item $N = \{i_1, ..., i_{3r}\}$ and $M = \cup_{k=1}^{3r}A_k$ where $A_1, ..., A_{3r}$ are  disjoint sets such that $|A_k| = a_k -1$ for each $1 \le k \le 3r$;
 \item $M(i_k) = A_k$ for each $1 \le k \le 3r$, and $p_i = 1$ for all $i \in N \cup M$.
\end{itemize}

Suppose that there exists a partition of $N$ into $m$ non-empty subsets $N^1$, ..., $N^m$ such that  
\begin{equation}\label{thm3: f<C=B}
 f(N^1, ..., N^m) \le C = B.
\end{equation}
For each $1 \le e \le m$, denote by $K_e$ the set of all $k$ such that $i_k \in N^e$.
Then, taking into account (\ref{f()}) and (\ref{thm3: f<C=B}),
\[
 mB \ge \sum_{e=1}^m \left(\sum_{i \in N^e}p_i + \sum_{j \in \cup_{i \in N^e} M(i)}p_j\right)  \ge \sum_{i \in N \cup M} p_i = \sum_{k=1}^{3r} \sum_{g \in A_k \cup \{i_k\}}p_g 
\]
\[= \sum_{k=1}^{3r} a_k = rB = mB,
\] 
and consequently, similarly to the proof of Theorem \ref{MWPSASM1}, for each $1 \le e \le m$,
\[
  B = \sum_{i \in N^e}p_i + \sum_{j \in \cup_{i \in N^e} M(i)}p_j = \sum_{k \in K_e} \sum_{g \in A_k \cup \{i_k\}}p_g = \sum_{k \in K_e}a_k.
\]
Hence, $Z_1 = K_1$, ..., $Z_r = K_r$ is a required partition of the set $\{1, 2, ..., 3r\}$.

Conversely, let $Z_1$, ..., $Z_r$ be a partition of the set $\{1, 2, \ldots, 3r\}$ such that each $Z_e$ satisfies (\ref{sum a=B}). Then, similarly to the proof of Theorem \ref{MWPSASM1}, taking into account (\ref{sum a=B}) and that $m = r$,  
\[
 N^1 = \cup_{k  \in Z_1}\{i_k\}, \hspace{0.5cm}..., \hspace{0.5cm}N^m = \cup_{k \in Z_r}\{i_k\}
\]
is a partition of $N$ which satisfies (\ref{thm3: f<C=B}). 
\qed
\end{proof}

\section{Approximation algorithms}
MWPSAS\_N1 can be viewed as a scheduling problem concerned with processing a finite set of jobs $J$ on $m > 1$ identical machines $M_1$, ..., $M_m$. In this scheduling problem, each $j \in J$ has an integer processing time $d_j$ and can be processed by any of the machines. A machine can process at most one job at a time. The goal is to minimise the makespan -- the time required for the completion of all jobs. Indeed, any instance of MWPSAS\_N1 can be viewed as an instance of the above scheduling problem with $J = N$ and $d_j = |M(j)| + 1$, for all $j \in J$. For this instance of the scheduling problem, the sets in a partition $N^1$, ..., $N^m$ in MWPSAS\_N1 specify what jobs should be assigned to  machines $M_1$, ..., $M_m$, respectively, and $f(N^1, ..., N^m)$ is the corresponding makespan. Therefore, any approximation algorithm for the above makespan minimisation problem is also applicable to MWPSAS\_N1.       

Consider the general case of MWPSAS. For any partition $R^1$, ..., $R^{r}$ of $N$, Approximation Algorithm below constructs a partition of $N$ into $m$ subsets $P^1$, ..., $P^m$. For this partition, Theorem \ref{A} establishes an upper bound on the deviation of $f(P^1, ..., P^m)$ from the optimal value of this function in terms of the partition $R^1$, ..., $R^{r}$. This leads to a two-phase optimisation procedure where the first phase is a choice of a partition $R^1$, ..., $R^{r}$ and the second phase is the conversion of $R^1$, ..., $R^{r}$ into $P^1$, ..., $P^m$. If $r = 1$ and consequently $R^1 = N$, i.e. no partition has been actually made, the algorithm below is applied to the original problem directly. Other two extreme choices of the initial partition are $R^1$, ..., $R^{r}$ where each $|R^k| = 1$ and consequently $r = |N|$, and $R^1$, ..., $R^{|M|}$ for MWPSAS\_M1 where each $R^k$ is $N(j)$ for some $j \in M$. It will be shown that, unless $P = NP$, the latter choice of the initial partition for MWPSAS\_M1 leads to the best possible polynomial-time  algorithm in terms of the guaranteed upper bound on the deviation from the optimal value of (\ref{f()}).    

Approximation Algorithm constructs the required partition using the value 
\begin{eqnarray*}
 D = \left\lceil \displaystyle  \frac{1}{m}\left(\sum_{j \in N}p_j + \sum_{u=1}^{r} \sum_{g \in \cup_{j\in R^u}M(j)}p_g \right)\right\rceil  \\ 
 + \max_{1 \le u \le r}\left\{
 \max_{j \in R^u}p_j + \sum_{g \in \cup_{j\in R^u}M(j)}p_g 
 \right\}-1
\end{eqnarray*}
which is computed based on the input partition $R^1$, ..., $R^{r}$.

\begin{algorithm}[H]
\renewcommand{\thealgorithm}{}
\floatname{algorithm}{}
\caption{Approximation Algorithm}
\begin{algorithmic}[1]
\label{algo-approx}
\STATE $P^1 = \emptyset$, $k = 1$, $e = 1$, $H^1 = R^1$, ..., $H^r = R^r$
\WHILE {$k \leq r$} \label{while k start}
 \STATE $H = H^k$
 \WHILE {$H \neq \emptyset$} \label{while H start}
  \STATE choose $i \in H$ \label{choose i}
  \IF {$\displaystyle \sum_{j \in P^e \cup \{i\}}p_j + \sum_{j \in \cup_{g\in P^e\cup \{i\}}M(g)}p_j \le D$} \label{if in H}
 	\STATE $P^e = P^e \cup \{i\}$
	\STATE $H^k = H^k\setminus \{i\}$
  \ENDIF
  \STATE $H = H\setminus \{i\}$
 \ENDWHILE \label{while H end}
 \IF {$H^k \neq \emptyset$} \label{if Hk}
  \STATE $e = e + 1$ \label{e=e+1}
  \STATE $P^e = \emptyset$
 \ELSE 
  \STATE $k = k+1$ \label{k+1}
 \ENDIF
\ENDWHILE \label{while k end}
\WHILE {$e < m$} \label{second while start}
\STATE $k = 1$
\WHILE {$|P^k| = 1$} 
\STATE $k = k + 1$ 
\ENDWHILE
\STATE choose $i \in P^k$
\STATE $P^k = P^k \setminus \{i\}$
\STATE $e = e + 1$
\STATE $P^e = \{i\}$
\ENDWHILE \label{second while end}
\RETURN $P^1$, ..., $P^m$ 
\end{algorithmic}
\end{algorithm}

\begin{theorem}\label{A}
Approximation Algorithm \ref{algo-approx} constructs a partition of $N$ into $m$ subsets, and for this partition $P^1$, ..., $P^m$, 
\[
 f(P^1, ..., P^m) - f^* 
\]
\[
 \le \left\lceil \displaystyle  \frac{1}{m}\left(\sum_{j \in N}p_j + \sum_{u=1}^{r} \sum_{g \in \cup_{j\in R^u}M(j)}p_g \right)\right\rceil + \max_{1 \le u \le r}\left\{
 \max_{j \in R^u}p_j + \sum_{g \in \cup_{j\in R^u}M(j)}p_g 
 \right\}
\]
\[
  - \max \left \{\left \lceil \frac{1}{m} \sum_{i \in N\cup M} p_i\right \rceil,  \; 
 \max_{i \in N}\left(p_i + \sum_{j \in M(i)}p_j \right)\right\} - 1
\]

where $f^*$ is the optimal value of (\ref{f()}).
\end{theorem}
\begin{proof}
Observe that if $P^e = \emptyset$, then the condition in line \ref{if in H} is satisfied which leads to the expansion of $P^e$ by $i$, chosen in line \ref{choose i}. Each time when $P^e$ is expanded by $i$, chosen in line \ref{choose i}, this $i$ is eliminated from $H^k$ and from $H$. Observe also that if $H^k \neq \emptyset$ after the execution of the while loop \ref{while H start} -- \ref{while H end}, then the next iteration of the while loop \ref{while k start} -- \ref{while k end} commences with this $H^k$ and $P^e = \emptyset$. Furthermore, index $k$ is increased according to line \ref{k+1}, and consequently $H^k$ is eliminated from further consideration, only when $H^k = \emptyset$. Hence, the while loop \ref{while k start} -- \ref{while k end} terminates after a final number of iterations with some partition $P^1$, ..., $P^{e'}$ of $N$.
 
In order to prove that $e' \le m$, assume to the contrary that $e' > m$. For each $1 \le e < e'$, denote by $k_e$ the largest among all $k$ such that $P^e\cap R^k \neq \emptyset$ and denote by $K_e$ the set defined as follows. If the increase of $e$ in line \ref{e=e+1} (and consequently the termination of the construction of $P^e$) is triggered in line \ref{if Hk} by set $H^{k_e + 1}$, then $K_e$ is the set of all $k$ such that $P^e\cap R^k \neq \emptyset$. Otherwise, if the increase of $e$ in line \ref{e=e+1} is triggered in line \ref{if Hk} by set $H^{k_e}$, then $K_e$ is the set of all $k$ such that $P^e\cap R^k \neq \emptyset$ and $k < k_e$. 

In the former case, i.e. in the case when $k_e \in K_e$, the while loop \ref{while H start} -- \ref{while H end}
failed to add to $P^e$ an element of $H^{k_e + 1}$. So, for any $i \in H^{k_e + 1}$,  
\[
 \sum_{j \in P^e}p_j + \sum_{k \in K_e}\sum_{j \in \cup_{g\in R^k}M(g)}p_j + p_i + \sum_{j \in M(i)}p_j \ge \sum_{j \in P^e \cup \{i\}}p_j 
 + \sum_{j \in \cup_{g\in P^e  \cup \{i\}}M(g)}p_j \ge D + 1
\]
and consequently 
\begin{equation}\label{thm: >D}
 \sum_{j \in P^e}p_j + \sum_{k \in K_e}\sum_{j \in \cup_{g\in R^k}M(g)}p_j \ge \frac{1}{m}\left(\sum_{j \in N}p_j + \sum_{u=1}^{r} \sum_{g \in \cup_{j\in R^u}M(j)}p_g \right).
\end{equation}
In the latter case, i.e. in the case when $k_e \notin K_e$, the while loop \ref{while H start} -- \ref{while H end} failed to add to $P^e$ the element $i \in R^{k_e}$, chosen as a result of the last execution of line \ref{choose i}. Hence,
\[
 \sum_{j \in P^e}p_j + \sum_{k \in K_e}\sum_{j \in \cup_{g\in R^k}M(g)}p_j + p_i + \sum_{j \in \cup_{g\in R^{k_e}}M(g)}p_j \ge \sum_{j \in P^e \cup \{i\}}p_j 
 + \sum_{j \in \cup_{g\in P^e  \cup \{i\}}M(g)}p_j 
\]
\[ 
 \ge D + 1,
\]
which again leads to (\ref{thm: >D}). Taking into account that, for any $1 \le u \le m$ and $1 \le v \le m$ such that $u \neq v$, $P^u \cap P^v = \emptyset$ and $K_u \cap K_v = \emptyset$, by adding (\ref{thm: >D}) for all $1 \le e \le m$,  
\[
 \sum_{e = 1}^m\sum_{j \in P^e}p_j \ge \sum_{j \in N}p_j + \sum_{u=1}^{r} \sum_{g \in \cup_{j\in R^u}M(j)}p_g - \sum_{e = 1}^m \sum_{k \in K_e}\sum_{j \in \cup_{g\in R^k}M(g)}p_j \ge \sum_{j \in N}p_j,
\]
which contradicts the assumption that $P^{m+1} \neq \emptyset$.

If $e' < m$, then the partition $P^1$, ..., $P^{e'}$, constructed by the while loop \ref{while H start} -- \ref{while H end}, becomes the initial current partition for the while loop \ref{second while start} -- \ref{second while end}. At each iteration,  the while loop \ref{second while start} -- \ref{second while end} constructs a new current partition of $N$ by finding a set in the current partition, containing more than one element (such a set exists by virtue of $m < |N|$), removing one element from this set, and introducing a new set that is comprised of only this one element.  

Taking into account that 
\[
 f^* \ge \max \left \{\left \lceil \frac{1}{m} \sum_{i \in N\cup M} p_i\right \rceil,  \; 
 \max_{i \in N}\left(p_i + \sum_{j \in M(i)}p_j \right)\right\} 
\]
and that $f(P^1, ..., P^m) \le D$,
\[
 f(P^1, ..., P^m) - f^* \le D - \max \left \{\left \lceil \frac{1}{m} \sum_{i \in N\cup M} p_i\right \rceil,  \; 
 \max_{i \in N}\left(p_i + \sum_{j \in M(i)}p_j \right)\right\}
\]
\[
 = \left\lceil \displaystyle  \frac{1}{m}\left(\sum_{j \in N}p_j + \sum_{u=1}^{r} \sum_{g \in \cup_{j\in R^u}M(j)}p_g \right)\right\rceil + \max_{1 \le u \le r}\left\{
 \max_{j \in R^u}p_j + \sum_{g \in \cup_{j\in R^u}M(j)}p_g 
 \right\}-1 
\]
\[
  - \max \left \{\left \lceil \frac{1}{m} \sum_{i \in N\cup M} p_i\right \rceil,  \; 
 \max_{i \in N}\left(p_i + \sum_{j \in M(i)}p_j \right)\right\}
\]
which completes the proof.
\qed
\end{proof}

\begin{corollary}
If, for MWPSAS\_M1, $r = |M|$ and each $R^u$ is $N(j)$ for some $j \in M$, then 
\[
 f(P^1, ..., P^m) - f^* \le 1
\] 
where $f^*$ is the optimal value of (\ref{f()}).
\end{corollary}

\section{Conclusions}
This paper considers a combinatorial optimisation problem, inspired by applications in software testing and manufacturing, and proves that even very restricted particular cases of this problem are NP-hard in the strong sense. These computational complexity results are complemented by a method of designing approximation algorithms and an upper bound on the deviation from the optimal value of the objective function which is valid for all algorithms constructed by this method. The results presented in the paper also contribute to the realm of scheduling with communication delay, since the studied combinatorial optimisation problem can be viewed as a parallel machine scheduling problem with the infinitely large communication delay and the partial order on the set of jobs in the form of a bipartite graph. Given the computational complexity of the considered problem, a logical continuation of the presented research would be the design and analysis of various approximation algorithms, including but not limited to the algorithms based on the method presented in the paper, as well as computational experiments with different exact algorithms.

\bibliographystyle{splncs04}
\bibliography{CommunicationDelay}

\end{document}